\documentclass{article}
\usepackage{ijcai23}

\usepackage{times}
\usepackage{soul}
\usepackage{url}
\usepackage[hidelinks]{hyperref}
\usepackage[utf8]{inputenc}
\usepackage[small]{caption}
\usepackage{graphicx}
\usepackage{amsmath}
\usepackage{amsthm}
\usepackage{amssymb}
\usepackage{booktabs}
\usepackage[ruled,vlined,linesnumbered]{algorithm2e}
\usepackage{hyperref}
\usepackage[nameinlink]{cleveref}
\usepackage{mathtools}
\usepackage{subcaption}
\usepackage{ wasysym }
\usepackage{comment}

\usepackage[switch]{lineno}

\theoremstyle{plain}

\newtheorem{theorem}{Theorem}
\newtheorem{lemma}{Lemma}
\newtheorem{corollary}{Corollary}

\newtheorem{observation}{Observation}

\def\arxiv{}

\title{Strategic Resource Selection with Homophilic Agents}

\author{
Jonathan Gadea Harder$^1$
\and
Simon Krogmann$^1$
\and
Pascal Lenzner$^1$\And
Alexander Skopalik$^2$
\affiliations
$^1$Hasso Plattner Institute, University of Potsdam\\
$^2$Mathematics of Operations Research, University of Twente\\
\emails
\{jonathan.gadeaharder, simon.krogmann, pascal.lenzner\}@hpi.de,
a.skopalik@utwente.nl
}

\newcommand{\game}{Schelling Resource Selection Game}
\newcommand{\gameshort}{SRSG}
\newcommand{\s}{\mathbf{s}}
\DeclareMathOperator*{\argmax}{arg\,max}

\ifdefined\arxiv
    \usepackage[style=authoryear,maxbibnames=99,dashed=false]{biblatex}
    \bibliography{ijcai23-submission}
    \renewcommand{\cite}[1]{\parencite{#1}}
\else
\fi

\begin{document}

\maketitle
\begin{abstract}
The strategic selection of resources by selfish agents is a classic research direction, with Resource Selection Games and Congestion Games as prominent examples. In these games, agents select available resources and their utility then depends on the number of agents using the same resources. This implies that there is no distinction between the agents, i.e., they are anonymous.

We depart from this very general setting by proposing Resource Selection Games with heterogeneous agents that strive for joint resource usage with similar agents. So, instead of the number of other users of a given resource, our model considers agents with different types and the decisive feature is the fraction of same-type agents among the users. More precisely, similarly to Schelling Games, there is a tolerance threshold $\tau \in [0,1]$ which specifies the agents' desired minimum fraction of same-type agents on a resource. Agents strive to select resources where at least a $\tau$-fraction of those resources' users have the same type as themselves. For $\tau=1$, our model generalizes Hedonic Diversity Games with a peak at $1$.

For our general model, we consider the existence and quality of equilibria and 
the complexity of maximizing social welfare. Additionally, we consider a bounded rationality model, where agents can only estimate the utility of a resource, since they only know the fraction of same-type agents on a given resource, but not the exact numbers. Thus, they cannot know the impact a strategy change would have on a target resource. Interestingly, we show that this type of bounded rationality yields favorable game-theoretic properties and specific equilibria closely approximate equilibria of the full knowledge setting.
\end{abstract}

\section{Introduction}
Selecting resources in a multi-agent setting is a long-established field of study in Artificial Intelligence, Operations Research, and Theoretical Computer Science. Resources can be as diverse as compute servers or printers, facilities like hospitals, universities, high schools or kindergartens, office rooms, restaurants or pubs, or driving routes to the workplace. The main common feature of all such resources typically is that the utility of the agents selecting them depends on the resource selections of all other agents. For almost all resources, the utility of an agent depends on the number of other agents that chose to select the same resource as the agent.

However, if the utility only depends on the number of agents that use a resource, this implies that the agents are indifferent to whom they are sharing their selected resource with. While this is a natural assumption in some settings, e.g., for jobs on a compute server or for customers of a gas station, there are many scenarios where real-world agents have more complex preferences. One prime example of this is the phenomenon of residential segregation~\cite{massey1988dimensions}. There, the agents are residents of a city that select a place to live. Typically, real-world residents are not indifferent to who else lives in their neighborhood, but instead prefer at least a certain fraction of neighbors of the same ethnic group or socioeconomic status. This behavior is commonly called homophily and it is seen as the driving force behind the strongly segregated neighborhoods, i.e., regions where similar people live, that can be observed in most major cities.

A classic way to model homophilic agents is to assume that agents have different types and that they strive for having at least a $\tau$-fraction of same-type agents sharing their selected resource. In terms of residential segregation, this is captured by Schelling's seminal agent-based model~\cite{SCHELLING}.

In this paper, we present and investigate a general model for strategic resource selection by a set of heterogeneous agents with homophilic preferences. For this, we incorporate Schelling's idea of threshold-based agent preferences into the classic setting of Resource Selection Games, where a set of resources is given and agents each can access a subset of them. Now, instead of selecting a resource with few users, an agent aims for selecting a resource that is shared at least with a $\tau$-fraction of users that have the same type as the agent.

We believe that this adds an important new dimension to classic agent-based resource selection problems. For example, our model allows us to investigate strategic school choice, where families of different ethnic groups that are located within a city select nearby schools for their children. Each family can select a subset of the available schools, typically schools that are in their neighborhood, but they want to ensure that their children have a certain fraction of fellow pupils of their own ethnic group.\footnote{Schools are more segregated than their neighborhoods~\cite{burgess2005parallel}, i.e., families actively select schools where the own ethnic group has a high presence~\cite{hailey22}.} Another example, is the choice of pubs that supporters of different sports teams select for watching the next match. For enjoying the match with friends, each supporter aims for patronizing a pub with a least a certain fraction of like-minded fans.

\begin{figure*}[ht]
\centering
\begin{subfigure}{0.30\textwidth}
  \centering
  \includegraphics[width=1\linewidth]{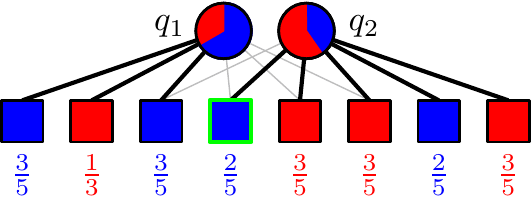}
  \caption{Social optimum. Neither impact-aware nor impact-blind equilibrium.}
  \label{fig:intro1}
\end{subfigure}
\hfill
\begin{subfigure}{0.30\textwidth}
  \centering
  \includegraphics[width=1\linewidth]{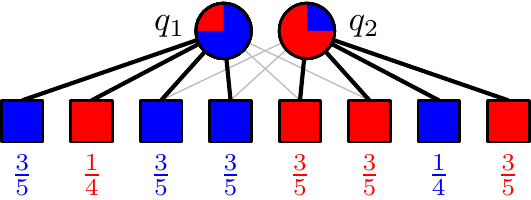}
  \caption{Impact-aware equilibrium. Not socially optimal.}
  \label{fig:intro2}
\end{subfigure}
\hfill
\begin{subfigure}{0.30\textwidth}
  \centering
  \includegraphics[width=1\linewidth]{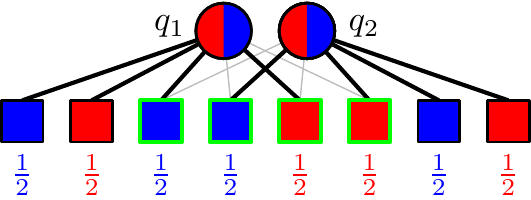}
  \caption{Impact-blind equilibrium but no impact-aware equilibrium.}
  \label{fig:intro3}
\end{subfigure}
\caption{Example instance of our model with three strategy profiles. The instance has two resources $q_1$ and $q_2$ (shown as circles, with their color fractions shown as pie charts), four red agents, and four blue agents, each shown as squares of the respective color. Moreover, we assume $\tau = \tfrac35$. Accessibility is shown via edges, thick black edges show the chosen resource of the respective agent. The fractions below the squares show the utilities of the agents. \textbf{(a)} shows the social optimum strategy profile with a social welfare of $\tfrac{62}{15} > 4.1$. It is neither an IAE nor an IBE, since the blue agent highlighted in green can increase her utility from $\tfrac25$ to $\tfrac35$ by selecting resource $q_1$ instead of $q_2$. \textbf{(b)} shows an IAE. Since it has social welfare of $4.1$ it is not socially optimal. \textbf{(c)} depicts an IBE with social welfare of $4$. It is not an IAE, since changing to the respective other resource is an impact-aware improving move for all the green highlighted agents.} 
\label{fig:intropics}
\end{figure*}

\subsection{Model and Notation}
We consider a strategic game, called the \emph{\game~(\gameshort)}, defined by a given bipartite \emph{accessibility graph} $G = (Q \cup A, E)$, with $Q \cap A = \emptyset$, where $Q$ is the set of \emph{resources} and $A$ is the set of \emph{agents} that want to use the resources at the same time.
Let $|Q|=k$, $|A|=n$, and $|E|=m$. 
An edge $\{q,a\} \in E$, with $q\in Q$ and $a\in A$ encodes that agent $a$ has access to resource $q$, i.e., resource $q$ is \emph{accessible} for agent $a$. We use $qa$ as shorthand for $\{q,a\}$. See \Cref{fig:intropics} for an illustration.

Most importantly for our paper, we assume that the society of agents is heterogeneous, in particular, that we have two\footnote{For more types of agents, \Cref{welopt1,thm:fip,poaupper,poalower,thm:ibe-pos,thm:pos-not-one} still hold.} types of agents that we distinguish by the colors \emph{red} and \emph{blue}. We have $A = R \cup B$, where $R$ is the set of $0<r<n$ many red agents and $B$ is the set of $b = n-r$ many blue agents.

We assume that all resources are identical in quality. 
Agents strategically select a single accessible resource to use, i.e., they each choose a single incident edge of $G$. We consider \emph{homophilic agents}, i.e., agents that favor a joint resource usage with other agents of the same type, if there are other agents using the same resource at all. In particular, analogously to Schelling's model for residential segregation~\cite{SCHELLING}, we assume a global tolerance threshold $\tau\in [0,1]$ for all the agents, that defines the minimum fraction of same-type agents that an agent seeks when using a resource. 
More formally, let $a\in A$ be any agent and let $X(a) \subseteq Q$ denote the set of \emph{accessible resources for agent $a$}, i.e., we have $X(a) = \{q\in Q \mid qa\in E\}$. Similarly, we define $Y(q) = \{a\in A \mid qa\in E\}$ for all $q \in Q$. Each agent $a$ selects a resource $\s(a) \in X(a)$ and we say that $\s(a)$ is agent $a$'s \emph{strategy}. The \emph{strategy profile} $\s = (\s(a_1),\dots,\s(a_n))$, with $a_i\in A$ and $\s(a_i) \in X(a_i)$, for $1\leq i \leq n$, then denotes the vector of strategies of all agents. 
For a strategy profile~$\s$ and a resource~$q$, we denote $\#(q,\s) = |\{a\in A \mid \s(a) = q\}|$ as the number of agents that jointly use $q$ under $\s$. Moreover, let $\#_R(q,\s) = |\{a\in R\mid \s(a)=q\}|$ denote the number of red agents that use resource $q$ under $\s$. Analogously, $\#_B(q,\s) = \#(q,\s) - \#_R(q,\s)$ is the number of blue agents using resource $q$ under $\s$.
Also, let \begin{linenomath}$$\rho_R(q,\s) = \frac{\#_R(q,\s)}{\#(q,\s)} \textnormal{\quad and \quad} \rho_B(q,\s) = \frac{\#_B(q,\s)}{\#(q,\s)}$$\end{linenomath} denote the fraction of red and blue agents, respectively, that use resource $q$ under strategy profile $\s$. Note that these fractions are undefined if no agent uses $q$ under $\s$.
 
The \emph{utility} of agent $a$ under strategy profile $\s$ is defined as 
\begin{linenomath}
\begin{equation*}
u(a,\s) = 
\begin{cases}
\min\{\rho_R(\s(a),\s),\tau\}, & \textnormal{if $a$ is red, and}\\
\min\{\rho_B(\s(a),\s),\tau\}, & \textnormal{if $a$ is blue.}
\end{cases}
\end{equation*}
\end{linenomath}
Thus, agents are striving for using a resource in a group of same-type agents that represents at least a $\tau$-fraction of all users of that resource.\footnote{Using a resource alone gives a fraction of same-type agents of~1.}
Note that $u(a,\s)$ is monotonically increasing in the fraction of same-type agents using the same resource. Also, if at least a $\tau$-fraction of same-type agents use resource $q$, then all of these agents have maximum utility.\footnote{Capping the utility at $\tau$ is a technical assumption that allows for more elegant proofs. All our results also hold if the utility is normalized to be $1$ if the fraction of same-type agents is at least $\tau$.}

Given a strategy profile $\s$ with $\s(a_i) = q$, we use the shorthand $\s = (q,\s_{-i})$, where vector $\s_{-i} = (\s(a_1),\dots,\s(a_{i-1}),\s(a_{i+1}),\dots,\s(a_n))$ is identical to the vector $\s$ without the $i$-th entry. Regarding strategy changes, we will consider two variants: impact-blind improving moves and impact-aware improving moves. Let $\s' = (q',\s_{-i})$ denote a strategy profile that is identical to $\s$ except for the $i$-th entry, i.e., $\s(a_j) = \s'(a_j)$ for $i\neq j$ and $\s(a_i) \neq \s'(a_i)$, with $\s'(a_i) = q'$. We say that the strategy change of agent $a_i$ from strategy $q$ to strategy $q'$ that leads from $\s$ to $\s'$ is an \emph{impact-aware improving move} if $u(a_i,(q',\s_{-i})) > u(a_i,(q,\s_{-i}))$.

The same strategy change is an \emph{impact-blind improving move}, if $\min\{\rho_R(q',\s), \tau\} > u(a_i,\s)$ and $a_i$ is red, or if $\min\{\rho_B(q',\s), \tau\} > u(a_i,\s)$ and $a_i$ is blue.
Since $\rho_R(q', \s)$ and $\rho_B(q', \s)$ are undefined for $\#(q', \s) = 0$, we say that switching to an empty resource is also an impact-blind improving move for an agent $a_i$, if $u(a_i, \s) < \tau$.
We call such an improving move impact-blind, since agent~$a_i$ only compares her utility with the fraction of same-type agents on resource~$q'$, prior to moving to resource~$q'$. This models that agent~$a_i$ has limited information about resource~$q'$, that is, she does not know $\#(q',\s_{-i})$ exactly, but only knows the fraction of agents of her type on $q'$ before the move. Hence, agent~$a_i$ cannot estimate the exact fraction of same-type agents using resource~$q'$ after she joins resource~$q'$, i.e., under strategy profile~$\s'$. She only knows that her joining resource~$q'$ does not decrease this fraction of same-type agents.

Based on impact-aware and impact-blind improving moves, we define two solution concepts for our strategic game. We say that $\s$ is in \emph{impact-aware equilibrium (IAE)} if no agent has an impact-aware improving move. Analogously, strategy profile $\s$ is in \emph{impact-blind equilibrium (IBE)} if no agent has an impact-blind improving move. Slightly abusing notation, we will use IAE and IBE also for denoting the sets of strategy profiles that are in IAE or IBE, respectively.
Since $u(a_i,(q',\s_{-i})) \geq \min\{\rho_R(q',\s), \tau\}$ if $a_i$ is red or $u(a_i,(q',\s_{-i})) \geq \min\{\rho_B(q',\s), \tau\}$ if $a_i$ is blue, every impact-blind improving move also is an impact-aware improving move, i.e., every strategy profile in IAE also is in IBE.
However, the converse does not hold.\footnote{Counterexample: Let two resources be used by exactly one red and one blue agent each.} Moreover, we will also consider approximate equilibria. A \emph{$\beta$-approximate IAE} is a strategy profile where no agent has an impact-aware improving move that yields $\beta$ times her current utility.

The \emph{social welfare} of a strategy profile $\s$, denoted by $W(\s)$, is defined as $W(\s) = \sum_{i=1}^n u(a_i,\s)$, i.e., as the sum of all the agents' utilities. For a given graph $G$, let $\textnormal{OPT}_G$, called \emph{social optimum}, denote a strategy profile that maximizes the social welfare for $G$. 
Using the social welfare, we now define our measure for the quality of equilibria. For a graph $G$, let $\textnormal{maxIAE}(G)$ and $\textnormal{minIAE}(G)$ be the strategy profiles in IAE for $G$ with maximum and minimum social welfare, respectively. Analogously, we define $\textnormal{maxIBE}(G)$ and $\textnormal{minIBE}(G)$ for IBE. 
Then the \emph{Price of Anarchy with respect to the IAE}, IAE-PoA for short, is defined as $\max_G\{\frac{W(\textnormal{OPT}_G)}{\textnormal{minIAE}(G)}\}$. The \emph{Price of Stability with respect to the IAE}, IAE-PoS for short, is defined as $\max_{G}\{\frac{W(\textnormal{OPT}_G)}{\textnormal{maxIAE}(G)}\}$. For the IBE, we define the IBE-PoA and the IBE-PoS analogously.

Regarding the game dynamics, if the agents only perform impact-aware improving moves, then we call this \emph{impact-aware dynamics}. Analogously, in \emph{impact-blind dynamics} only impact-blind improving moves occur. If starting from every strategy profile every sequence of impact-aware (impact-blind) improving moves is finite, then we say that the \emph{impact-aware (impact-blind) finite improvement property}, short IA-FIP or IB-FIP, holds.

\subsection{Related Work}
Resource selection as an optimization problem is a classic topic in combinatorial optimization~\cite{papadimitriou1998combinatorial}, with many variants of scheduling, packing, or covering problems as examples.

Strategic resource selection started with Congestion Games~\cite{rosenthal}, where a set of resources is given, and agents select a subset of them. The agents' cost then depends solely on the number of users on their chosen set of resources. In extensions, weighted versions and even agent-specific cost functions are allowed~\cite{milchtaich1996congestion}. Prominent examples are the strategic selection of paths in a network~\cite{roughgarden2002bad,pos} or server selection by selfish jobs~\cite{vocking2007selfish}. 
Also, competitive facility location, where the facilities compete for the clients~\cite{vetta-utility-system} or the clients compete for facilities~\cite{load-balancing,Peters2018,krogmann,krogmannAAAI} can be seen as strategic resource selection. 
Another example is the selection of group activities~\cite{darmann2012group,igarashi2017group}. In all the above models, the utility of an agent depends on the number of other agents that select the same resources.

In contrast to this, and much closer to our work, are models with heterogeneous agents. Recently, game-theoretic models for the creation of networks by homophilic agents have been studied~\cite{BLM22}. More related to our model is Schelling's model for residential segregation~\cite{SCHELLING}, where agents with a type strategically select a location in a residential area. These agents behave according to a threshold-based utility function that yields maximum utility if at least a certain fraction of same-type agents populate the neighborhood of the selected location. Game-theoretic variants of Schelling's model, called Schelling Games, have recently been studied~\cite{CLM18,E+19} and also variants became popular, where agents strive to maximize the fraction of same-type neighbors~\cite{A+21,bullinger21,KKV21,BBLM22,KKV22} or with single-peaked preferences~\cite{BBLM22singlepeaked,single_peaked_jump}. Schelling Games are different from our model since every resource, i.e., location, can only be chosen by at most one agent and thus the respective neighborhoods of agents only partially overlap. Moreover, the size of these neighborhoods is bounded and given by the graph that models the residential area, whereas in our model there is no limit on the number of users of a specific resource.

Closest related to our model are Hedonic Diversity Games~\cite{BEI19,BE20,Darmann21,Ganian22}, where agents of different types strategically form coalitions, and their utility depends on the fraction of same-type agents in their chosen coalition. While preferences over fractions may be arbitrary in these games, our model generalizes the special case with preferences resembling a $\tau$-threshold function as in Schelling's model or with agents having single-peaked utilities with a peak at $1$, since in our model the access to the resources can be restricted. Thus, these special cases of Hedonic Diversity Games match our model with a suitable $\tau$ on a complete bipartite accessibility graph. Also related are Hedonic Expertise Games~\cite{CaskurluKO21}, where the utility of agents increases with more different types in the coalition.

\subsection{Our Contribution}
We consider strategic resource selection problems with homophilic agents. In contrast to previous work, the utility of our agents does not depend on the number of other agents that use the same resource, but on the fraction of same-type agents. This opens up a new direction for resource selection problems, like Congestion Games. 

Another main conceptual contribution is the study of impact-blind equilibria, which can be understood as a natural bounded-rationality variant of classic Nash equilibria and hence might be of independent interest. Impact-blindness models that the exact number of users of a resource might not be known to the agents, but only the fraction of user types. For example, for a large-scale multi-agent technical system, it might be reported that it is currently used by 25\% red and 75\% blue users. From an agent's perspective, this is a game with incomplete information. For risk-averse agents, the best response is indeed equivalent to being impact-blind. Even for risk-neutral agents, if the (expected) number of agents is large, the optimal behavior is impact-blind. Mathematically speaking, as the number of agents grows, the set of Bayesian equilibria of the game with incomplete information converges to the set of impact-blind equilibria.

As our main technical contribution, we consider strategic resource selection with homophilic agents both with impact-aware and with impact-blind agents. For the latter, we prove that equilibria always exist and that they can be constructed efficiently. Moreover, equilibria are guaranteed to exist for $\tau \leq \frac12$ for impact-aware agents. Also, we show that specific impact-blind equilibria resemble 2-approximate impact-aware equilibria, which ensures the existence of almost stable states for any $\tau$ in the impact-aware setting.

Regarding the quality of equilibria, we prove tight constant bounds on the PoA for both versions and we show that the PoS is $1$ for $\tau=1$. On the complexity side, we show that maximizing social welfare is NP-hard in general, but efficient computation is possible for restricted instances.

\ifdefined\arxiv

\else
See~\cite{homophilic_arxiv} for all omitted details.
\fi

\section{Complexity}\label{sec:complexity}

We begin by studying the computational complexity of finding optimal profiles of \game{}s. We show that this is a hard problem in general, but for sparse instances, i.e., with bounded degrees of the resource nodes or the agent nodes, we can compute optimal solutions efficiently. 

\begin{theorem}
\label{welopt1}
    For any threshold $\tau>0$, it is {\sc NP}-hard to decide if every agent can get maximum utility.
\end{theorem}

\ifdefined\arxiv
\begin{proof}
We show this by reducing from $(3,4)$-SAT, which is 3SAT where each variable appears in at most $4$ clauses. This has been shown to be {\sc NP}-complete~\cite{TOVEY198485}.

Given a $(3,4)$-SAT instance $\phi$ with variables $X_1,\ldots,X_m$ and clauses $C_1,\ldots,C_r$, we construct a \game~$\Gamma_\phi$ as follows.

For each variable $X_i$, we introduce two {\em variable resources} $x_i$ and $ \neg x_i$ and $v:=2\cdot\lceil \frac{4}{\tau}\rceil$ {\em variable agents} $\chi_{i,1},\ldots,\chi_{i,v}$ which are colored red and connected to both variable resources $x_i$ and $\neg x_i$.
For each clause $C_j$, we have a blue {\em clause agent} $\gamma_j$ that is connected to the three variable resources that correspond to the literals of clause $C_j$.

We first show that if $\phi$ is satisfiable, there is a monochromatic strategy profile that gives all agents maximum utility.
Let $X^*_1,\ldots,X^*_m$ be a satisfying assignment for $\phi$, we then assign the variable agents as follows. If $X^*_j$ is true, the agents $\chi_{j,1},\ldots,\chi_{j,v}$ are assigned to resource $\neg x_j$. Otherwise, they are assigned to resource $x_j$.
We assign every clause agent to a variable resource of a variable that satisfies that clause. As $\phi$ is satisfying such variables exist and by the choice above there is no variable agent on that resource.
In this strategy profile, every agent is on a monochromatic resource.

Now, if $\phi$ is not satisfiable, let $\s$ be an arbitrary strategy profile of $\Gamma_\phi$. We show that there exists at least one unsatisfied agent. We derive a variable assignment $X'_1,\ldots,X'_n$ from $\s$ as follows. If at least $\frac12 v$ variable agents are on resource $x_i$, we set $X'_i$ to false; otherwise, $X'_i$ is set to true.
Since $\phi$ is not satisfiable, there exists a clause $C_j$ which is not satisfied by $
X'$. For the corresponding blue clause agent $\gamma_j$, all three resources that she can choose from have at least $\frac12 v$ many red variable agents. As each variable appears in at most $4$ clauses, there are at most $3$ other (blue) clause agents on each of the variable resources. Hence, the ratio for $\gamma_j$ is at most 
$\frac{4}{4+\frac12 v} = \frac{4}{4+\lceil \frac{4}{\tau}\rceil} < \tau $.

Therefore, there is a strategy profile in $\Gamma_\phi$ that gives every agent maximum utility if and only if $\phi$ is satisfiable.
\end{proof}
\fi

As a corollary, we obtain hardness for the problem of computing a profile maximizing social welfare
since finding an optimal profile also solves the problem of deciding whether every agent can achieve utility $\tau$.
\begin{corollary}
    Computing the social optimum is {\sc NP}-hard.
\end{corollary}

\Cref{welopt1} even holds if we restrict it to instances where every agent can choose from a set of at most three resources.
However, this can not be extended to a maximum of two.
\begin{theorem}
For $\tau=1$, deciding if every agent can get utility $1$ is solvable in polynomial time if each $a\in A$ has degree $2$.
\end{theorem}
\ifdefined\arxiv
\begin{proof}
    We show this by constructing a 2SAT formula $\varphi$ which is true if and only if every agent can get utility $1$. Note that satisfying all agents implies that no two agents with different colors are assigned to the same resource.

    For each resource $q$, we introduce a variable $x_q$. We will interpret $x_q=$ true if $q$ is only used by red agents and $x_q=$ false if $q$ is only used by blue agents.
    
     For each red agent $a \in R$ with two choices $q_1$ and $q_2$, i.e., $(q_1,a)\in E$ and $(q_2,a)\in E$, we add the clause $(x_{q_1} \vee x_{q_2})$ to $\varphi$. If a red agent has only one available resource $q$, we add the clause $(q)$.
     For each blue agent $a \in R$ with two choices $q_1$ and $q_2$, i.e. $(q_1,a)\in E$ and $(q_2,a)\in E$, we add the clause $(\neg x_{q_1} \vee \neg x_{q_2})$ to $\varphi$. If a blue agent has only one available resource $q$, we add the clause $(\neg q)$.
    
    It is easy to see that $\varphi$ is satisfiable if and only if there is a monochromatic assignment $s$ of agents to resources: If there is a satisfying solution to $\varphi$, we assign each red agent to a resource with the corresponding variable set to true. Each blue agent is assigned to a resource with the corresponding variable set to false. By construction of $\varphi$, such resources exist.
    For the other direction: from a monochromatic assignment $s$, we can derive a satisfying assignment for $\varphi$ by setting each $x_q$ to true if and only if $q$ is used by red agents only.
    
    Finally, we use that 2SAT can be solved efficiently.
\end{proof}
\fi

A similar result holds if resource nodes have low degree.
\begin{theorem}
    A social optimum can be computed in polynomial time for $\tau \in [0,1]$ if the degree of each resource is $2$.
\end{theorem}
\begin{proof}
 
The following procedure yields an optimal strategy profile:
In the first step, we iteratively check for a resource that is accessible by only a single agent or by agents of the same color only. We assign those agents to that resource and remove the resource and those agents from the instance.
We are left with an instance in which every resource is accessible by exactly one red and one blue agent in its associated accessibility graph. 
In the second step, we compute a maximum matching in the accessibility graph and assign agents to resources according to the matching. Each remaining unmatched agent is assigned to an arbitrary accessible resource.

To prove optimality, we first observe that all agents assigned in step one have maximal utility and all resources removed in step one are not accessible by any agent left for step two.
It remains to show that the assignment of step two is optimal. To that end, let $n$ be the number of agents of the instance at the beginning of step two and let $k$ denote the cardinality of a maximum matching. Hence, our algorithm computes a profile with $n-k$ resources that each have two agents of different colors and $2k-n$ resources with exactly one agent. Assume this was not optimal, then there needs to be a solution with fewer than $n-k$ resources that have two differently colored agents. Hence, the total number of used resources needs to be larger, as $n$ agents still have to be assigned. However, this implies the existence of a matching of cardinality strictly larger than $k$ which yields a contradiction.

The algorithm can be implemented in polynomial time using a standard algorithm, e.g. \cite{hopcroft}.
\end{proof}

\section{Equilibrium Existence and Computation}\label{sec:existence}
In this section, we show that IBE exist in all instances and for all $\tau > 0$, since an ordinal potential function exists.
Improving response dynamics converge in a polynomial number of steps, but an even more efficient greedy algorithm exists to construct IBE directly from scratch.
\subsection{Impact-Blind Equilibria}
\begin{lemma}
\label{lem:social-welfare-increase}
For $\tau = 1$, an impact-blind improving move increases social welfare.
\end{lemma}
\begin{proof}
Let, w.l.o.g., a red agent make an impact-blind improving move from resource $q$ to $q'$, changing the strategy profile from $\s$ to $\s'$.
Let $r_1 = \#_R(q, \s)$, $b_1 = \#_B(q, \s)$, $r_2 = \#_R(q', \s)$ and $b_2 = \#_B(q', \s)$.

The total social welfare of the agents on $q$ in $\s$ is $W_q(\s) = \frac{r_1^2 + b_1^2}{r_1 + b_1} = r_1 - b_1 + 2 \left(\frac{b_1^2}{r_1 + b_1}\right)$ and in $\s'$ it is $W_q(\s') = \frac{(r_1-1)^2 + b_1^2}{r_1 + b_1 - 1} = r_1 - b_1 - 1 + 2 \left(\frac{b_1^2}{r_1 + b_1 - 1}\right)$.
If $r_1 + b_1 - 1 = 0$, then $a$ would not improve by switching to $q'$.
Thus, the change of welfare of the agents on $q$ is
$W_q(\s') - W_q(\s) = 2 \left(\frac{b_1^2}{r_1 + b_1 - 1}\right) - 2 \left(\frac{b_1^2}{r_1 + b_1}\right) - 1 = 2\left(\frac{b_1^2}{(r_1 + b_1 - 1)(r_1 + b_1)}\right)- 1$.

Similarly, the total social welfare of the agents on $q'$ in $\s$ is $W_{q'}(\s) = \frac{r_2^2 + b_2^2}{r_2 + b_2} = r_2 - b_2 + 2 \left(\frac{b_2^2}{r_2 + b_2}\right)$ and in $\s'$ it is $W_{q'}(\s') = \frac{(r_2+1)^2 + b_2^2}{r_2 + b_2 + 1} = r_2 - b_2 + 1 + 2 \left(\frac{b_2^2}{r_2 + b_2 + 1}\right)$.
Thus, the change of welfare of the agents on $q$ is
$W_{q'}(\s') - W_{q'}(\s) = 2 \left(\frac{b_2^2}{r_2 + b_2 + 1}\right) - 2 \left(\frac{b_2^2}{r_2 + b_2}\right) + 1 = -2\left(\frac{b_2^2}{(r_2 + b_2 + 1)(r_2 + b_2)}\right)+ 1$.

Therefore, the total difference is
$W(\s') - W(\s) = 2\left(\frac{b_1^2}{(r_1 + b_1 - 1)(r_1 + b_1)} -\frac{b_2^2}{(r_2 + b_2 + 1)(r_2 + b_2)}\right) > 2\left(\frac{b_1^2}{(r_1 + b_1)^2} -\frac{b_2^2}{(r_2 + b_2)^2}\right)$.
Since the move is improving, we have $\frac{b_1}{r_1+b_1} > \frac{b_2}{r_2+b_2}$ and therefore $W(\s') - W(\s) > 0$.
\end{proof}

Note that \Cref{lem:social-welfare-increase} does not hold for impact-aware improving moves. (See \Cref{thm:pos-not-one}.) With the difference considered in the proof of \Cref{lem:social-welfare-increase}, we can also bound the number of steps needed to reach an IBE.
\begin{lemma}
\label{polytimesteps}
    For $\tau = 1$, an IBE is reached after $O(n^5)$ impact-blind improving moves.
\end{lemma}
\begin{proof}
As seen in \Cref{lem:social-welfare-increase}, an improving move increases social welfare by 
$W(\s') - W(\s) > 2\left(\frac{b_1^2}{(r_1 + b_1)^2} -\frac{b_2^2}{(r_2 + b_2)^2}\right) = 2\frac{b_1^2(r_2 + b_2)^2 - b_2^2(r_1 + b_1)^2}{(r_1 + b_1)^2(r_2 + b_2)^2}$.

Also by \Cref{lem:social-welfare-increase}, the numerator is a positive integer, so $W(\s') - W(\s) > \frac{1}{n^4}$.
Since the social welfare is in $[0,n]$, its maximum is reached after at most $n^5$ steps.
\end{proof}

Since an impact-blind improving move for a threshold $\tau < 1$ is also an impact-blind improving move for $\tau = 1$ and since we can compute an improving response in $O(m)$, we get:
\begin{theorem}
\label{thm:fip}
For any $\tau \in [0,1]$, the \game{} possesses the IB-FIP and an IBE can be computed with runtime $O(n^5m)$.
\end{theorem}
Note that the social welfare at $\tau=1$ is always an ordinal potential function, independently of the actual value of $\tau$.%
\ifdefined\arxiv
{ }However, social welfare is not a potential function for the impact-aware setting.
\begin{observation}
For $\tau = 1$ an impact-aware improving move may decrease social welfare.
\end{observation}
\begin{proof}
Let $a$ be a red agent using a resource $q$ with $\rho_R(q, \s) = \frac{19}{100}$ and $\#(q,s) = 100$.
Let there be an improving move for $a$ to a resource $q'$ with $\rho_R(q, \s) = \frac{1}{10}$ and $\#(q',s) = 10$.
Then the change in social welfare is $\frac{18^2+81^2}{99} + \frac{2^2+9^2}{11} - \frac{19^2+81^2}{100} - \frac{1^2+9^2}{10} = -\frac{81}{550}$.
\end{proof}
\else%
\footnote{In the full version~\cite{homophilic_arxiv}, we show that for $\tau<1$ the social welfare is not necessarily a potential function.}
\fi

We provide \Cref{alg:greedy-eq} to greedily compute an IBE, i.e., we can circumvent finding equilibria via expensive improving response dynamics.
However, we still think that the IB-FIP is important as the agents can also find an equilibrium independently.
Intuitively, our algorithm checks which resource can achieve the highest red fraction, taking into account the blue agents that have only one resource available.
Then it assigns all possible red agents and the necessary blue agents to that resource and removes it from the instance.
\begin{algorithm}[t]
    \caption{computeEquilibrium(G)}
    \label{alg:greedy-eq}
    \While{$Q$ not empty}{
        assign all $a \in B$ with $|X(a)| = 1$\;
        $q \gets \argmax_{q \in Q}{\frac{|R \cap Y(q)|}{|B \cap \text{assigned}(q)|}}$\;
        assign all $a \in R \cap Y(q)$ to $q$\;
        remove $q$ and its assigned agents from $G$\;
    }
\end{algorithm}
We show, that the algorithm removes the resources sorted by their fraction of red agents in the resulting equilibrium.

\begin{lemma}
\label{lem:removal-order}
    If \Cref{alg:greedy-eq} producing the IBE $\s$ removes resource $q_1$ before resource $q_2$, then $\rho_R(q_1, \s) \geq \rho_R(q_2, \s)$.
\end{lemma}
\begin{proof}
    While running the algorithm, for each resource $q$, the number of assigned blue agents to $q$ monotonically increases and the number $|R \cap Y(q)|$ of red agents assignable to $q$ monotonically decreases.
    Thus, after the removal of $q_1$, no resource $q_2$ with a higher fraction can be removed.
\end{proof}

\begin{theorem}
    For $\tau \in [0,1]$, \Cref{alg:greedy-eq} computes an IBE in runtime $O((m + k) \log k)$.
\end{theorem}
\begin{proof}
    For the runtime, we build a data structure in which for each resource $q$ we store the number of assigned blue agents and the number of red agents with an edge to $q$, with runtime $O(m+k)$.
    Updating this data structure and finding the blue agents can be done in amortized $O(m)$ time, as we only have to do an operation for each edge that is removed from the instance.
    Additionally, we maintain a Fibonacci heap~\cite{fibonacci-heap} of resources, to select the next resource to be removed.
    This needs $O(m)$ decrease-key operations ($O(\log{k})$) for each of the updates of the aforementioned data structure.
    For extracting the $k$ resources we need $k$ extract-max operations ($O(\log{k})$).

    A red agent does not want to change strategies, as she cannot access the resources removed before her assignment and, by \Cref{lem:removal-order}, all other resources have a worse fraction. For blue agents, the argument is analogous. 
\end{proof}
Note that \Cref{alg:greedy-eq} is not correct for IAE.%
\ifdefined\arxiv
\begin{observation}
\Cref{alg:greedy-eq} is incorrect for IAE.
\end{observation}
\begin{proof}
Let $Q=\{q_1, q_2\}$ and $\tau=1$.
We have one red agent $a_r$ connected to $q_1$ and 3 red agents connected to $q_2$.
We also have one blue agent $b_r$ connected to both $q_1$ and $q_2$, while another blue agent is connected to only $q_2$.
The algorithm removes $q_1$ first with only $a_r$ assigned.
All other agents are assigned to $q_2$.
In this state, however, $b_r$ can make an impact-aware improving move by switching to $q_1$ changing her utility from $\frac25$ to $\frac12$.
\end{proof}
\else%
\footnote{See~\cite{homophilic_arxiv} for a counterexample.}
\fi

\subsection{Impact-Aware Equilibria}
For impact-awareness, we show the existence of an equilibrium for $\tau \leq \frac12$ by using a potential argument.
First, we give an ordinal potential function that always remains constant or increases with an improving move.
The function is the sum of majority sizes over all resources.
\begin{lemma}
\label{lem:majority-potential}
    With an impact-aware improving move in the \game{}, the potential function $\Phi(\s) = \sum_{q \in Q}{\max\{\#_R(q, \s), \#_B(q, \s)\}}$ never decreases, for all possible values of $\tau$. (However, it is possible that it does not change.)
    The number of steps increasing $\Phi(\s)$ in a sequence of improving moves is limited.
\end{lemma}
\begin{proof}
The potential $\Phi$ can only have integer values in $[0, n]$, limiting the number of increasing moves.
Let, w.l.o.g, $a$ be a red agent making an improving move from resource $q$ to $q'$ changing the state from $\s$ to $\s'$.
We study the possible cases for the relation between $\#_R$ and $\#_B$ at $q$ and $q'$ and consider the terms of $\Phi$ for $q$ and $q'$ as all other terms do not change.

\noindent\textbf{Case 1:} ($\#_R(q, \s) > \#_B(q, \s)$):
We have $\Phi_q(\s)=\max\{\#_R(q, \s), \#_B(q, \s)\} = \max\{\#_R(q, \s'), \#_B(q, \s')\} + 1=\Phi_q(\s')+1$.
Since the move is improving, $\#_R(q', \s') > \#_B(q', \s')$ holds and therefore $\Phi_{q'}(\s)=$ $\max\{\#_R(q', \s), \#_B(q', \s)\} = \max\{\#_R(q', \s'), \#_B(q', \s')\}$ $- 1=\Phi_{q'}(\s')-1$.
Thus, the value of $\Phi$ remains unchanged.

\noindent\textbf{Case 2:} ($\#_R(q, \s) = \#_B(q, \s)$):
We have $\Phi_q(\s)=\max\{\#_R(q, \s), \#_B(q, \s)\} = \max\{\#_R(q, \s'), \#_B(q, \s')\}=$ $\Phi_q(\s')$.
Since the move is improving, $\#_R(q', \s') > \#_B(q', \s')$ and thus $\Phi_{q'}(\s)=\max\{\#_R(q', \s), \#_B(q', \s)\} =$  $ \max\{\#_R(q', \s'), \#_B(q', \s')\} - 1=\Phi_q'(\s') - 1$.
Thus, the value of $\Phi$ increases by 1.

\noindent\textbf{Case 3:} ($\#_R(q, \s) < \#_B(q, \s)$ and $\#_R(q', \s) < \#_B(q', \s)$):
Blue agents stay in the majority for this move in $q$ and $q'$, so $\Phi$ remains unchanged.

\noindent\textbf{Case 4:} ($\#_R(q, \s) < \#_B(q, \s)$ and $\#_R(q', \s) \geq \#_B(q', \s)$):
We have $\Phi_q(\s)=\max\{\#_R(q, \s), \#_B(q, \s)\} = \max\{\#_R(q, \s'), \#_B(q, \s')\}=\Phi_q(\s')$ and $\Phi_{q'}(\s)=\max\{\#_R(q', \s), \#_B(q', \s)\} = \max\{\#_R(q', \s'), \#_B(q', \s')\} $ $- 1=\Phi_q'(\s')-1$.
Thus, the value of $\Phi$ increases by 1.
\end{proof}
With \Cref{lem:majority-potential}, we know that in a sequence of improving moves, Cases 2 and 4 of the proof occur only finitely often.

\begin{theorem}
    For $\tau \leq \frac12$, the \game{} has the IA-FIP.\footnote{We conjecture that this also holds for arbitrary $\tau$.}
\end{theorem}
\begin{proof}
Let $Z(\s)$ be the descendingly sorted vector of utilities in $\s$.
We show that $Z$ in combination with $\Phi$ is a lexicographic potential function for the $\gameshort$.

To prove this, let, w.l.o.g., a red agent $a$ make an improving move from resource $q$ to $q'$ changing the state from $\s$ to $\s'$.
We consider three cases for $\rho_R(q', \s)$ before the move:

\noindent\textbf{Case 1 ($\rho_R(q', \s) < \frac12$):}
The only agents that lose utility are the red agents that keep using $q$.
(The utility of the blue agents using $q'$ does not fall below $\frac12$ and therefore does not change.)
Since agent $a$ ends up with a utility greater than what the losing agents had before, the earliest change in the vector is an increase.

\noindent\textbf{Case 2 ($\rho_R(q', \s) = \frac12$):}
In this case, the potential given in \Cref{lem:majority-potential} increases (Case 2), hence this case may only occur a finite number of times in a sequence of improving moves.

\noindent\textbf{Case 3 ($\rho_R(q', \s) > \frac12$):}
Only blue agents at $q'$ and red agents at $q$ lose utility.
All of them have a utility strictly smaller than $\frac12$ in $\s'$ before the move.
Since $a$ improves from a utility smaller than $\frac12$ to equal to $\frac12$, the first change in the vector is an improvement in the new spot of $a$.

The number of possible values of an entry in $Z$ is limited, and thus also the number of possible values for the vector.
\end{proof}

\section{Equilibrium Approximation}

Next, we show that we can compute an approximate IAE by using an IBE with specific properties as a proxy.
\begin{theorem}
    A $2-$approximate impact-aware equilibrium can be computed in runtime $O(n^5m)$ for any $\tau \in [0,1]$.
\end{theorem}
\begin{proof}
For the proof, we use $\tau=1$, but a $2-$approximate IAE for $\tau=1$ is also a $2-$approximate IAE for all other values of $\tau$ as the utility gain of a move monotonically decreases with decreasing $\tau$.
Let, w.l.o.g., a red agent $a$ make an impact-aware improving move from resource $q$ to $q'$, changing the strategy profile from an IBE $\s$ to $\s'$.
First, we show that if $a$ improves by a factor of more than 2, then social welfare increases.
Let $r_1 = \#_R(q, \s)$, $b_1 = \#_B(q, \s)$, $r_2 = \#_R(q', \s)$ and $b_2 = \#_B(q', \s)$.
As $\s$ is an IBE and therefore $a$'s current utility is in the interval of $\frac{r_2}{r_2+b_2}$ and $\frac{r_2+1}{r_2+b_2+1}$, we have $r_2=0$ to allow for an improvement factor of more than 2.
The move changes the sum of utilities of agents using $q'$ by the value
\[
    \frac{1}{b_2+1}-b_2\left(\frac{1}{b_2+1}\right) = \frac{2}{b_2+1}-1
\]
For $q$ this change is
\begin{linenomath}
    \begin{align*}
    &(b_1-r_1+1)\left(\frac{r_1}{r_1+b_1}-\frac{r_1-1}{r_1+b_1-1}\right)-\frac{r_1}{r_1+b_1}\\
    =&(b_1-r_1+1)\cdot\frac{b_1}{(r_1+b_1)(r_1+b_1-1)}-\frac{r_1}{r_1+b_1}\\
    =&\frac{(b_1+r_1+1)b_1 -2b_1r_1}{(r_1+b_1)(r_1+b_1-1)}-\frac{r_1}{r_1+b_1}\\
    =&1-\frac{r_1}{r_1+b_1}-\frac{2b_1r_1}{(r_1+b_1)(r_1+b_1-1)}-\frac{r_1}{r_1+b_1}\\
    \geq&1-4\cdot\frac{r_1}{r_1+b_1}.
    \end{align*}
\end{linenomath}
In the last step, we use $\frac{b_1}{r_1+b_1-1} \leq 1$. Therefore the total change in social welfare is at least
$
\frac{2}{b_2+1}-4\cdot\frac{r_1}{r_1+b_1}
$,
which is greater than 0, by the assumption that 
    $2\left(\frac{r_1}{r_1+b_1}\right) < \frac{1}{b_2+1}$.

For finding an IBE with no possible strategy change that increases social welfare, we use the algorithm in \Cref{polytimesteps}.
However, we execute all strategy changes that improve social welfare and not just impact-blind improving moves.
\end{proof}

\section{Equilibrium Quality}\label{sec:quality}
In this section, we compute the exact Price of Anarchy for all values of $\tau$ which holds for both impact-aware and impact-blind agents.
Additionally, we give an upper bound on the impact-blind Price of Stability and show that for $\tau \leq \frac12$ the social optimum is not necessarily an IBE.

\subsection{Price of Anarchy}
For giving an upper bound on the Price of Anarchy, we first prove lower bounds for the sum of utilities of the agents using a resource. After that, we use the fact that a social welfare optimum assigns at most utility $\tau$ to every agent for deriving our upper bound. Finally, we give a class of instances in which we match this upper bound on the PoA asymptotically.
This gives us an exact PoA for both impact-blind and impact-aware agents.
\begin{lemma}
\label{boundlemma1}
Consider a resource $q$ in a arbitrary state $\s$, with, w.l.o.g., $\#_R(q, \s) \geq \#_B(q, \s)$.
If $\rho_R(q, \s) \geq \tau$, then the sum of utilities of the agents using $q$ is at least $\#(q, \s)\left(\tau-\frac{\tau^2}{4}\right)$.
\end{lemma}
\begin{proof}
Let $\rho_B(q, \s) = \tau - \epsilon$, with $\epsilon \geq 0$.
The sum of utilities of agents using $q$ is
$S(\epsilon) = \#(q, \s)((1 - (\tau - \epsilon))\tau + (\tau - \epsilon)^2) = \#(q, \s)(\tau - \tau\epsilon + \epsilon^2)$.
To find the minimum of this expression, we set $\frac{d}{d\epsilon}S(\epsilon) = 0$ and get $0 = \#(q, \s)(-\tau + 2\epsilon) \to \epsilon = \frac{\tau}{2}$.
Thus, the minimum for the sum of utilities $S(\epsilon)$ is $S(\frac{\tau}{2}) = \#(q, \s)\left(\tau - \frac{\tau^2}{4}\right)$.
\end{proof}

\begin{lemma}
\label{boundlemma2}
Consider a resource $q$ in a arbitrary state $\s$ with, w.l.o.g., $\#_R(q, \s) \geq \#_B(q, \s)$.
If $\rho_R(q, \s) \leq \tau$, then the sum of utilities of the agents using $q$ is at least $\frac{\#(q, \s)}{2}$.
\end{lemma}
\begin{proof}
Form pairs between blue and red agents using $q$.
Each pair has a combined utility of 1.
All unpaired red agents have a utility of at least $\frac12$.
Thus, the average utility of an agent using $q$ is at least $\frac12$.
\end{proof}
\begin{theorem}\label{poaupper}
The IAE-PoA and IBE-PoA for $\tau \leq 2-\sqrt{2}$ are at most $\frac{4}{4-\tau}$ and for $\tau \geq 2-\sqrt{2}$ they are at most $2\tau$.
\end{theorem}
\begin{proof}
Consider a resource $q$ in a arbitrary state $\s$ with, w.l.o.g., $\#_R(q, \s) \geq \#_B(q, \s)$ and consider the sum of utilities $S$ of the agents using $q$ in two different cases.

\noindent\textbf{Case 1 ($\rho_R(q, \s) \geq \tau$):}
    By \Cref{boundlemma1}, we have $S \geq \#(q, \s)\left(\tau-\frac{\tau^2}{4}\right)$.
        
\noindent\textbf{Case 2 ($\rho_R(q, \s) \leq \tau$):}
    By \Cref{boundlemma2}, we have $S \geq \frac{\#(q, \s)}{2}$.
    
Now, we check which case dominates for which $\tau$:
\[
\frac{1}{2}\leq \tau-\frac{\tau^2}{4} \iff 2-\sqrt{2} \leq \tau \leq 2+\sqrt{2}\text.
\]
Thus, for $\tau \geq 2-\sqrt{2}$, the state $\s$ with the lowest social welfare has at least $W(\s) \geq \frac{n}{2}$.
Similarly, for $\tau \leq 2-\sqrt{2}$, the lowest social welfare is at least $W(\s) \geq n\left(\tau-\frac{\tau^2}{4}\right)$.
The highest social welfare any state can have is $\tau n$, so we get bounds of $\frac{4}{4-\tau}$ and $2\tau$ for the PoA, respectively.
\end{proof}

We match these upper bounds asymptotically, by creating instances in which only a constant number of agents do not use resources with the worst-case distributions given in \Cref{boundlemma1,boundlemma2}.

\begin{theorem}
\label{poalower}
The IAE-PoA and IBE-PoA for $\tau \leq 2-\sqrt{2}$ are $\frac{4}{4-\tau}$ and for $\tau \geq 2-\sqrt{2}$ they are $2\tau$.
\end{theorem}
\begin{proof}
For $\tau \geq 2-\sqrt{2}$, we create the instance $G_\alpha$ with $Q=\{q_1,q_2,q_3\}$ with a parameter $\alpha$ such that $\alpha \in \mathbb{N}$.
We have a set $R_x$ of $\frac{\alpha(2-\tau)}{2}$ red agents with edges to $q_1$ and $q_3$ and a set $B_x$ of $\frac{\alpha \tau}{2}$ blue agents with edges to $q_1$ and $q_2$.
(Since we let $\alpha\to\infty$ later, we can use the nearest integer values for $R_x$ and $B_x$ with an error that goes to $0$.)
Furthermore, we have a set of red agents $R_z$ and blue agents $B_z$, both with $\lceil\frac2\tau\rceil$ agents and edges to $q_2$ and $q_3$.
Assigning all red agents to $q_3$ and assigning all blue agents to $q_2$ gives all agents the maximum utility of $\tau$, achieving the same best-case optimum we used in \Cref{poaupper}. See \Cref{fig:poasub2}.

\begin{figure}[t]
\centering
\begin{subfigure}{0.45\linewidth}
  \centering
  \includegraphics[width=1\linewidth]{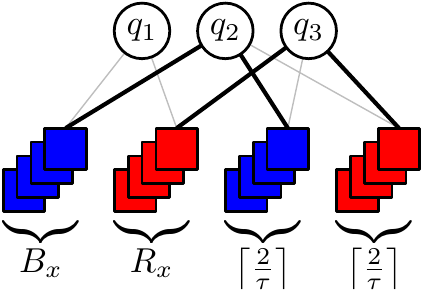}
  \caption{Social optimum.}
  \label{fig:poasub2}
\end{subfigure}
\hfill
\begin{subfigure}{0.45\linewidth}
  \centering
  \includegraphics[width=1\linewidth]{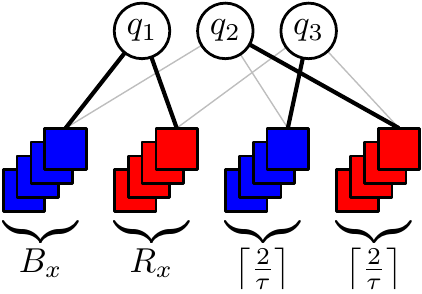}
    \captionsetup{justification=centering}
  \caption{Socially bad IAE.}
  \label{fig:poasub1}
\end{subfigure}
\caption{Example instance from the proof of \Cref{poalower} showing that the PoA bound in \Cref{poaupper} is tight. 
}
\label{fig:poasharp}
\end{figure}

Assigning $R_x$ and $B_x$ to $q_1$, while assigning $R_z$ and $B_z$ to $q_2$ and $q_3$ respectively yields an IAE (see \Cref{fig:poasub1}), as the agents using $q_1$ can only switch to a resource used by only the respective other color.
In particular, a blue agent $a \in B_x$ with utility $\frac\tau2$ cannot improve by moving to $q_3$ with utility $\frac{1}{\lceil\frac2\tau\rceil + 1} < \frac\tau2$ (and similar for agents in $R_x$). 
Since the agents on $q_1$ have the worst-case distribution of \Cref{boundlemma1}, we asymptotically achieve the PoA bound with $\alpha\to\infty$.

For $\tau \leq 2-\sqrt{2}$, we use the same construction, but with $|R_x|=|B_x|=\frac\alpha2$, which achieves the worst-case distribution of \Cref{boundlemma2}.
\end{proof}

\subsection{Price of Stability}

As seen in \Cref{lem:social-welfare-increase}, the IBE-PoS is 1 for $\tau=1$. We now generalize this for arbitrary $\tau$.
\begin{theorem}
    \label{thm:ibe-pos}
    The IBE-PoS for arbitrary $\tau>0$ is at most $\frac{1}{\tau}$.
\end{theorem}
\begin{proof}
For a threshold $\tau$, let $\s_\tau$ be the social welfare optimum and denote by $W_\tau(\s)$ the social welfare of a state $\s$.
Observe that by \Cref{lem:social-welfare-increase} $\s_1$ is an impact-blind equilibrium for arbitrary $\tau$. 
Since the utility of an agent decreases at most by factor $t$, when changing the threshold from $\tau=1$ to $\tau=t$, we have $W_{t}(\s_1) \geq t W_1(\s_1)$.
Also $W_1(\s_1) \geq W_{1}(\s_t) \geq W_{t}(\s_t)$, and, therefore, we have $\frac{W_t(\s_t)}{W_t(\s_1)} \leq \frac{1}{t}$.
\end{proof}
Note that for $\tau < \frac1{\sqrt{2}}$, this result is dominated by the bound on the IBE-PoA.
Combining these two bounds, we get a bound of $\sqrt{2}$ on the IBE-PoS for any $\tau>0$.

\begin{theorem}
\label{thm:pos-not-one}
    The IBE-PoS for $0 < \tau \leq\frac{1}{2}$ is larger than $1$.
\end{theorem}
\ifdefined\arxiv
\begin{proof}
     Let $\frac{x}{y} \le \tau $ with $x,y \in \mathbb{N}$ and $x\ge6$.
    Consider an instance with two resources $q_1$ and $q_2$.
    Resource $q_1$ is used by $x-1$ blue agents and $(y-x)+1$ red agents, $q_2$ is used by one blue agent. All agents have only one accessible resource except agent $a$. Agent $a$ is one of the blue agents on $q_1$ and also has an edge to $q_2$. See \Cref{fig:posnotone}.
    \begin{figure}[t]
\centering
\begin{subfigure}{0.45\linewidth}
  \centering
  \includegraphics[width=1\linewidth]{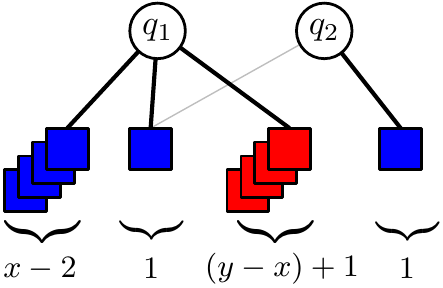}
    \captionsetup{justification=centering}
  \caption{Social optimum.}
  \label{fig:posopt}
\end{subfigure}
\hfill
\begin{subfigure}{0.45\linewidth}
  \centering
  \includegraphics[width=1\linewidth]{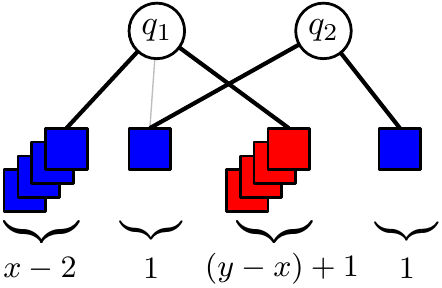}
    \captionsetup{justification=centering}
  \caption{The unigue IBE.}
  \label{fig:posIBE}
\end{subfigure}
\caption{Example instance from the proof of \Cref{thm:pos-not-one} showing that the PoS is larger than $1$.
}
\label{fig:posnotone}
\end{figure}

Note that in this game there are only two strategy profiles. Let the profile with $a$ on $q_1$ be $OPT$ and the one with $a$ on $q_2$ is the unique IBE.
    For the latter, observe that agent $a$ increases her utility by $\frac{1}{y}$ when switching to $q_2$.

    To quantify the loss in social welfare, observe that the red agents are in the majority on $q_1$ before and after the move, hence, their utility does not change. For each of the remaining $x-2$ blue agents on $q_1$ the utility decreases by $\frac{y-x+1}{y(y-1)}$ if $a$ switches to $q_2$. Thus the total decrease of the utility of agents on $q_1$ is
    \begin{linenomath}
    \begin{align*}
            \frac{(x-2)(y-x+1)}{y(y-1)}&=
            \frac{xy-2y+3x-x^2-2}{y(y-1)}\text.
    \end{align*}
    Subtracting the social welfare increase by the blue moving agent $a$ gives a total decrease in social welfare of
    \begin{align*}
    \frac{xy-2y+3x-x^2-2}{y(y-1)} - \frac{y-1}{y(y-1)}\geq \frac{x}{y(y-1)}>0 \text,
    \end{align*}
    \end{linenomath}
where the first inequality follows from $x\geq 6$ and $\frac{x}{y} \le \frac{1}{2}$, hence, $xy \geq 3y +x^2$.

Therefore $OPT$ is indeed the optimal solution but not an IBE and the theorem follows.
\end{proof}
\fi

\section{Conclusion and Outlook}
In this paper, we introduce \game{}s in which agents strategically select resources to optimize the fraction of same-type agents on their respective selected resources.
We investigate agents that fully know the impact of their own actions and also agents that are blind to their exact impact.
For the second case, which is arguably more realistic for settings with limited information, we present an efficient algorithm for computing equilibria. Moreover, we show that specific impact-blind equilibria approximate impact-aware equilibria well. Also, we show that the PoA is at most 2, showing that the game is well-behaved.
We believe that impact-blind equilibria are a natural object to study also in other settings, like Schelling Games, Hedonic Games, and other types of Resource Selection Games.
 
Given our algorithmic advances on the problem, it may now be interesting to investigate real-world examples of the problem.
Especially in the case of school choice, there may be data available on which our model could be used to draw additional conclusions.
With further research, it may be possible to identify measures that help to desegregate schools, which is a pressing current problem.
Note that careful consideration must be applied to our measure of social welfare.
Even though it accurately depicts the agents' preferences, the social optimum yields maximally segregated usage of the resources, which might be undesirable for some applications, e.g., in schools.
For these cases, diversity is a more suitable measure of social welfare.

\ifdefined\arxiv
    \printbibliography
\else
    \bibliographystyle{named}
    \bibliography{ijcai23-submission}
\fi

\end{document}